\documentclass[letterpaper, 10 pt, journal, twoside]{IEEEtran}

\IEEEoverridecommandlockouts

\usepackage{color}
\usepackage{enumerate,graphicx,epstopdf}
\usepackage{amsmath,amssymb,bm,amsthm}
\usepackage{cite}
\usepackage{mathtools}
\usepackage{array}
\usepackage{algpseudocode,algorithm,algorithmicx}
\usepackage{color}
\usepackage{booktabs}
\usepackage[hyphens]{url}
\usepackage{lipsum}

\usepackage[dvipsnames]{xcolor}
\usepackage[deletedmarkup=sout,authormarkup=superscript]{changes}

\definechangesauthor[name={Dominic}, color=RedViolet]{DLM}
\definechangesauthor[name={Dominic}, color=OliveGreen]{EB}

\newcommand{\ints}{\mathbb{Z}}
\newcommand\intrng[2]{\mathbb{Z}_{[#1,#2]}}
\newcommand{\reals}{\mathbb{R}}
\newcommand{\mc}{\mathcal}

\newcommand{\la}[0]{\langle}
\newcommand{\ra}[0]{\rangle}
\newcommand{\mto}[0]{\rightrightarrows}
\newcommand{\argmin}{\mathop{\rm argmin}}

\newtheorem{ass}{Assumption}

\newtheorem{lmm}{Lemma}
\newtheorem{thm}{Theorem}
\newtheorem{rmk}{Remark}

\title{On Robustness in Optimization-Based \\Constrained Iterative Learning Control 
\thanks{The authors are with the ETH Zürich Automatic Control Laboratory, Physikstrasse 3, 8092 Zürich, Switzerland. A. Rupenyan is also with Inspire AG. Email:\texttt{\{dliaomc, ebalta, ralisa, jlygeros\}@ethz.ch}. This research is supported by the Swiss National Science Foundation through NCCR Automation (Grant Number 180545).}}
\author{Dominic Liao-McPherson, Efe C. Balta, Alisa Rupenyan and John Lygeros}

\begin{document}

\maketitle
\thispagestyle{empty}

\begin{abstract}
Iterative learning control (ILC) is a control strategy for repetitive tasks wherein information from previous runs is leveraged to improve future performance. Optimization-based ILC (OB-ILC) is a powerful design framework for constrained ILC where measurements from the process are integrated into an optimization algorithm to provide robustness against noise and modelling error. This paper proposes a robust ILC controller for constrained linear processes based on the forward-backward splitting algorithm. It demonstrates how structured uncertainty information can be leveraged to ensure constraint satisfaction and provides a rigorous stability analysis in the iteration domain by combining concepts from monotone operator theory and robust control. Numerical simulations of a precision motion stage support the theoretical results.
\end{abstract}

\begin{IEEEkeywords}
Optimization algorithms, Robust control, Iterative learning control, Manufacturing systems and automation
\end{IEEEkeywords}

\section{Introduction}

Iterative learning control (ILC) is a feedforward control technique for repetitive tasks that seeks to iteratively improve performance by learning from previous trials. The primary objectives in ILC are rejection of repetitive disturbances and compensation for plant-model mismatch. ILC is performant, often converging adequately within just a few iterations, robust, and straightforward to implement \cite{bristow2006survey}. As such, it has been successfully applied in a variety of areas, for example in robotics, manufacturing, and chemical processing~\cite{ahn2007iterative}. 

There is an extensive body of literature on ILC of unconstrained systems. However, all real systems are subject to constraints e.g., actuator saturations, and ILC algorithms may fail if constraints are ignored \cite{chen1999learning}. ILC was originally formulated as an unconstrained tracking problem \cite{uchiyama1978formation,arimoto1984bettering}, which poses challenges in handling constraints in a systematic manner. Optimization-based ILC (OB-ILC) provides a powerful and complimentary framework for the design and analysis of ILC policies \cite{gunnarsson2001design, owens2005iterative,mishra2010optimization,chen2021iterative,owens2016ilcbook}, especially when constraints are present. 

In typical OB-ILC, the ILC learning policy is defined via the solution of a data dependent optimization problem and ILC is implemented by alternating between obtaining data from the system and updating the input using the policy. The designer's task is to construct this problem so that the policy has desirable transient behaviour e.g., monotonicity of the tracking error in a specific norm~\cite{barton2010norm,chu2010iterative,amann1996iterative}, and its fixed point coincides with the control objectives. A more systematic approach is to encode the control objectives in an optimization problem involving the unknown process and construct an ILC policy by incorporating data from the process into a suitable algorithm, similar to online/feedback optimization methods \cite{hauswirth2021optimization,simonetto2020time}.

An interior point based OB-ILC algorithm is presented in \cite{mishra2010optimization}. Methods based on the successive projection algorithm are presented in \cite{chu2010iterative} and \cite{chen2019generalized} for the input and input/output constrained cases respectively. Sequential quadratic programming and norm-optimal type algorithms are proposed in \cite{schollig2009optimization} and \cite{amann1996iterative}, respectively. 
However, the constrained convergence analyses provided in these works only cover the nominal case without process noise or model mismatch.

Sufficient conditions for robustness of unconstrained gradient descent based and norm-optimal ILC are derived in \cite{owens2009robust} and \cite{barton2010norm} respectively. Sufficient conditions for stability under structured uncertainty are derived in \cite{moon1998robust,van2009iterative} for the unconstrained case and in \cite{son2015robust} for the input constrained case using robust control techniques. More recently, \cite{chen2021iterative} provides results for constrained robust convergence under the assumption of inactive constraints and without non-repetitive noise. In \cite[Chapter 12.3]{owens2016ilcbook}, constrained convergence results are presented for two cases where: 1) the constraints are met exactly, but system measurements are not leveraged, leading to a non-robust controller, 2) system measurements are utilized in the ILC update, but the constraints are only asymptotically met. A constrained nonlinear programming approach with $0$-th order corrections is presented in \cite{baumgartner2020zero} and includes convergence results with approximate sensitivities, providing difficult-to-verify sufficient conditions and without noise. In a separate non-ILC context, input constrained feedback optimization of power systems 
is presented in \cite{colombino2019towards}, but output constraints, nor noise are considered.Therefore, there is a gap in the existing literature for treating input/output constrained robust convergence with constraint satisfaction guarantees in the presence of active constraints, non-repeating process noise, and modeling error.

In this paper, we investigate OB-ILC of noisy and constrained linear processes using forward-backward splitting (FBS) methods. We combine ideas from robust constrained control and monotone operator theory to study the constraint satisfaction and convergence properties of our proposed policy. Our contributions are threefold:
\begin{enumerate}[i)]
\item We present a methodology for handling input and output constraints in the presence of noise and model mismatch.
\item We illustrate how a process model can be interpreted as a preconditioner and explicitly analyze the effect of modelling error on the convergence rate of the ILC policy.
\item We derive computable sufficient conditions for input-to-state stability of the closed-loop in the iteration domain.
\end{enumerate}

\textit{Notation:} For $x,y\in \reals^n$, $A\in \reals^{n\times m}$, and a positive definite matrix $P = P^T \succ 0$ the weighted inner product is $\la x,y\ra_P = x^T P y$, $\|x\|_P^2 = \la x,x\ra_P$ and $\|A\|_P$ is the induced matrix norm. The convex hull is denoted by $\mathrm{co}$. The normal cone mapping for a closed convex set $\mc{C}\subseteq \reals^n$ is $\mc{N}_\mc{C}(\xi) = 
  \{w~|~\langle \xi - u,w\rangle \geq 0,~\forall u \in \mc C\}$ if $ \xi\in \mc C$ and empty otherwise. For $\mc{U}\subseteq \reals^n$ and $\mc{V}\subseteq \reals^n$ the Pontryagin set difference is $\mc{U} \ominus \mc{V} =\{x\in \mc{U}~|~x+y \in \mc{U}~\forall y\in \mc{V}\}$. The radius of a set $\mc{D}$ containing $0$ is $\mathrm{rad}~\mc{D} = \min~r$ such that $\mc{D} \subseteq \{x~|~\|x\|_\infty\leq r\}$.

\section{Problem Setting}
Consider a repetitive process governed by the mapping
\begin{equation} \label{eq:mapping}
y = Gu + w + d
\end{equation}
where $u \in \reals^n$ is the input, $y\in \reals^m$ is the output, $G\in \reals^{m\times n}$ and $w\in \reals^m$ are parameters, and $d$ is a disturbance drawn from a known compact set $\mc{D} \subset \reals^m$. The parameters $\theta = [G~w]$ are iteration invariant while $d$ varies with iteration. 

We wish to enforce the input and output constraints
\begin{equation} \label{eq:constraints}
	u\in \mc{U} \text{ and } y\in \mc{Y},
\end{equation}
where $\mc{U}\subset \reals^n$ and $\mc{Y}\subseteq \reals^m$ are closed, convex, and non-empty sets. A wide variety of practical constraints can be encoded using \eqref{eq:constraints}, such as actuator saturations, actuator rate limits, or tooltip velocity limits, see e.g., \cite{mishra2010optimization} or \cite{chu2010iterative}. 

The mapping parameters $\theta = [G~w]$ are unknown, but the output $y$ corresponding to an input $u$ can be measured by \textit{running an experiment}. We assume that the true mapping lies within a known uncertainty set $\Theta \subset \reals^{m \times (n+1)}$. In this paper, we employ the following convex hull representation\footnote{This kind of uncertainty set is commonly generated by e.g., set-membership identification or interval estimation \cite{milanese2013bounding}.} of $\Theta$.
\begin{ass}\label{ass:uncertainty}
The parameters $\theta = [G~w]$ of the repetitive process \eqref{eq:mapping} satisfy $\theta \in \Theta = \textrm{co}~\{\theta_1,\ldots,\theta_N\}$ where $\theta_i = [G_i~~w_i]~~i \in \intrng{1}{N}$ are known.
\end{ass}

Our objective is to drive the output of \eqref{eq:mapping}, to a reference output $r\in\reals^m$ while satisfying \eqref{eq:constraints}. The ILC approach to this problem is to construct an iterative process of the form
\begin{subequations} \label{eq:closed-loop}
\begin{align} \
  y_k &= Gu_k + w + d_k, \label{eq:mapping_a}\\
  u_{k+1} &= \mc{T}(u_k,y_k), \label{eq:update_law}
\end{align}
\end{subequations}
and design the ILC policy $\mc{T}: \reals^n \times \reals^m \to \mc{U}$ so that $ y_k \approx r$ as $k\to \infty$. 

In the ILC literature \eqref{eq:mapping} is known as the lifted system representation and can encode a variety of processes. For example, it can represent finite trajectories of the linear system
\begin{subequations} \label{eq:LTV_system}
\begin{align}
x(i+1) &= A(i) x(i) + B(i) u(i),~~x(0) = x_0\\
y(i) &= C(i) x(i) + c(i)
\end{align}
\end{subequations}
as seen in \cite{barton2010norm,amann1996iterative}, additive manufacturing processes governed by partial differential equations \cite{hoelzle2015spatial}, or transfer function based representations \cite{mishra2010optimization}. In this example, the term $Gu$ in \eqref{eq:mapping} is the forced response of \eqref{eq:LTV_system}, $w$ is the response to the initial condition $x_0$, and $d$ represents the measurement noise $c$.

\section{Optimization Based ILC Policy Design}
This paper adopts an optimization-based framework for ILC. If $d$ is absent and $G$ and $w$ are known then one can address our problem by solving
\begin{subequations} \label{eq:opt}
\begin{align} 
&\underset{y\in \mc{Y},u\in \mc{U}}{\mathrm{min}} \quad \frac12||y - r||_Q^2 + \frac12\|u\|_R^2 \\
&~~~\mathrm{s.t.}\quad \quad y = Gu + w,
\end{align}
\end{subequations}
where $Q = Q^T \succeq 0$ and $R = R^T \succeq 0$ are weighting matrices. This formulation can be used to encode a variety of meaningful objectives as illustrated by the example in Section~\ref{ss:simulations}.

The optimization problem \eqref{eq:opt} can be written compactly as
\begin{equation} \label{eq:opt_compact}
\underset{u\in \mc{Z}}{\mathrm{min}} ~~ \phi(u) = \frac12 u^T H u + f^T u 
\end{equation}
where $H = G^T Q G + R$, $f = G^T Q(w-r)$ and $\mc{Z} = \{u\in \mc{U}~|~Gu+w \in \mc{Y}\}$. Under convexity assumptions on $\phi$ and $\mc{Z}$, the following Variational Inequality (VI) is necessary and sufficient for optimality of \eqref{eq:opt_compact}
\begin{equation} \label{eq:ideal_vi}
 \nabla \phi(u) + \mc{N}_{\mc{Z}}(u)\ni 0
\end{equation}
where $\mc{N}_{\mc{Z}} :\mc{Z} \mto \reals^{n}$ is the normal cone mapping of $\mc{Z}$. This VI can be readily solved using a variety of methods, we focus on the forward-backwards splitting (FBS) algorithm
\begin{gather}\label{eq:FBsplitting}
u_{k+1} = \Pi_\mc{Z}^W(u_k - \alpha W^{-1} \nabla \phi(u_k)),
\end{gather}
where $W = W^T \succ 0$ is a \textit{preconditioner}, $\alpha > 0$ is a step-size,  and
\begin{equation} \label{eq:proj_def}
	\Pi_\mc{Z}^W(x) = (I + \alpha W^{-1} \mc{N}_{\mc{Z}})^{-1}(x) = \argmin_{v\in \mc{Z}} \|v-x\|_W^2
\end{equation}
is the $W$-weighted projection. Given exact knowledge of $G$ and $w$ the FBS converges to the unique minimizer of \eqref{eq:opt_compact} for appropriately chosen $W$ and $\alpha$, see e.g., \cite{bauschke2011convex}. Unfortunately, the FBS algorithm \eqref{eq:FBsplitting} is not directly implementable as it requires perfect knowledge of $G$ and $w$. 

Next, we illustrate how to modify the FBS to design ILC policies for constrained systems. The idea is to integrate measurements from the system into an optimization algorithm to provide robustness against disturbances and modelling errors. Since the true $G$ and $w$ are not available, we make the following modifications to the FBS algorithm \eqref{eq:FBsplitting} to make it suitable for online use:
\begin{enumerate}[i)]
\item We select an approximate model $M\approx G$ satisfying $M \in \mathrm{co}~\{G_1,\ldots,G_N\}$ and replace $G^T$ with $M^T$ in the gradient $\nabla \phi(u) = G^TQ(Gu + w -r) + Ru$
\item We incorporate the physical process into the algorithm by replacing the term $Gu + w$ in $\nabla \phi(u)$ with measurements $y$ from the system;
\item We replace $\mc{Z}$ with the tightened set
\begin{equation*} 
	\mc{X} = \{u\in \mc{U}~|~Gu + w + d \in \mc{Y}~~\forall [G~w]\in \Theta, d\in \mc{D} \},
\end{equation*}
to ensure constraint satisfaction;
\item We precondition the algorithm using the model by setting $W = M^T Q M + R$.
\end{enumerate}
These modifications to the forward-backwards algorithm result in the following optimization-based ILC policy
\begin{equation} \label{eq:ILC_policy}
u_{k+1} = \mc{T}(u_k,y_k) :=\Pi_\mc{X}^W (u_k - \alpha W^{-1}\bar F(u_k,y_k)),
\end{equation}
where $\bar F(y,u) = M^TQ(y - r) + Ru$. Under this preconditioner selection and in the absence of constraints, the algorithm (\ref{eq:ILC_policy}) is identical to the norm-optimal ILC algorithms in the literature, e.g., \cite{barton2010norm}, with the addition of the step-size parameter $\alpha$ and without an input variation penalty in the cost function.

The following assumptions ensure that the optimization problem \eqref{eq:opt} is well-posed and that the policy \eqref{eq:ILC_policy} is well-defined and implementable.
\begin{ass} \label{ass:mega-assumption}
The following hold: (i) $H = G^TQG + R \succ 0$; (ii) $\mc{X}$ is closed, convex, and non-empty; (iii) $\mc{D}$ is known, convex and compact; (iv) $W = M^TQM + R \succ 0$.
\end{ass}
The condition (i) is uncheckable (it requires knowledge of the true system) but can be ensured by setting $R \succ 0$. To check (ii), note that due to Assumption~\ref{ass:uncertainty} $\mc{X}$ reduces to 
\begin{equation} \label{eq:Xdef}
	\mc{X} = \{u\in \mc{U}~|~G_i u + w_i \in (\mc{Y} \ominus \mc{D}),~~i\in \intrng{1}{N}\}
\end{equation}
which is closed and convex if $\mc{U}$ and $\mc{Y}$ are closed and convex. Non-emptiness can then be checked via a feasibility problem\footnote{The Pontryagin-difference of convex sets can be computed using well established tools e.g., the MPT3 toolbox \cite{herceg2013multi}.}. The remaining points involve only known problem data.

The ILC policy \eqref{eq:ILC_policy} can re-expressed as a convex optimization problem and efficiently implemented using standard software. Using the projection equations \eqref{eq:proj_def}, the ILC policy satisfies the following inclusions
\begin{subequations}
\begin{gather}
	u_{k+1} + \alpha W^{-1}\mc{N}_\mc{X}(u_{k+1}) \ni u_k  - \alpha W^{-1} \bar F(u_k,y_k),\\
	\Leftrightarrow W(u_{k+1} - u_k) + \alpha \bar F(u_k,y_k) + \mc{N}_\mc{X}(u_{k+1})\ni 0. \label{eq:vi_weighted_proj}
\end{gather}
\end{subequations}
Let $u_{k+1} = \mc{T}(u_k,y_k)$ as in \eqref{eq:ILC_policy}, then the inclusion \eqref{eq:vi_weighted_proj} is necessary and sufficient for optimality of the following convex program with parameters $u = u_k$ and $y = y_k$
\begin{subequations} \label{eq:backwards_opt}
\begin{align}
\mc{T}(u,y) = &\argmin_{v\in \mc U}\frac12 \|v - u\|_W^2 + \alpha v^T \bar F(u,y)\\
&~~~\mathrm{s.t.} \quad G_i v + w_i \in (\mc{Y} \ominus \mc{D}),~ i \in \intrng{1}{N}.
\end{align}
\end{subequations}
Algorithm~\ref{algo:FBS-ILC} summarizes the implementation of the proposed ILC policy with a cost function based stopping criterion.

 \begin{algorithm}[H]
\caption{FBS based ILC}
\label{algo:FBS-ILC}
\begin{algorithmic}[1]
\Require $u_0$, $\epsilon> 0$, $k_{max} > 0$
\State $u_0 \gets \Pi_\mc{X}^W(u_0)$, $k \gets 0$,
\Repeat
  \State Apply $u_k$ to the process \eqref{eq:mapping} and measure $y_k$.
  \State  Solve the problem \eqref{eq:backwards_opt} to compute $u_{k+1} = \mc{T}(u_k,y_k)$.
  \State $k\gets k+1$
\Until{$\|u_{k+1} - u_k\|} \leq \epsilon$ or $k > k_{max}$
\end{algorithmic}
\end{algorithm}

Implementing \eqref{eq:backwards_opt} requires solving potentially large scale convex optimization problems. If $G$ arises from a linear dynamical system or is derived from a spatial process as in \cite{hoelzle2015spatial} then efficient solvers for model predictive control problems or Fourier transform based methods  can be leveraged.

\section{Analysis}
This section investigates the properties of the closed-loop system \eqref{eq:closed-loop}. Because of the disturbances, we mix monotone operator theory with the system theoretic input-to-state stability framework\cite{jiang2001input}. To begin the analysis, we substitute \eqref{eq:mapping_a} into \eqref{eq:update_law} leading to the following noisy forward-backward splitting algorithm
\begin{equation} \label{eq:analysis_mapping}
	u_{k+1} = \Pi^W_\mc{X}[u_k - \alpha W^{-1}(\tilde F(u_k) + M^TQ d_k)]
\end{equation}
where $\tilde F(u) = \tilde H u + \tilde f$, $\tilde H = R+ M^T Q G$, and $\tilde f = M^TQ(w-r)$. We impose the following conditions on the mapping $\tilde F$ and the step length $\alpha$.
\begin{ass} \label{ass:convergence-assumptions}
Define $\mu = \min_{i\in \intrng{1}{N}}\mu_i$, $L = \max_{i\in \intrng{1}{N}}L_i$ where $L_i$ and $\mu_i$ are the largest and smallest Eigenvalues of the matrix $\Xi_i = 0.5 W^{-\frac12}(H_i + H_i^T)W^{-\frac12}$, respectively and $H_i = M^TQG_i + R$. Then the following hold: (i) $\mu > 0$; (ii) $\alpha\in (0,2\mu/L^2)$.
\end{ass}
If the original problem \eqref{eq:opt} is strongly convex then $\mu> 0$ is guaranteed to hold if the uncertainty set $\Theta$ is sufficiently small. As such, Assumption~\ref{ass:convergence-assumptions}.(i) states that we cannot expect convergence for an arbitrarily bad model drawn from an arbitrarily large uncertainty set. Since $\Theta$ is known, $\mu$ and $L$ can be computed numerically using an Eigenvalue solver and $\alpha$ can be chosen such that Assumption~\ref{ass:convergence-assumptions}.(ii) holds.

We first show that Assumptions~\ref{ass:uncertainty}-\ref{ass:convergence-assumptions} are sufficient for constraint satisfaction and robust stability of \eqref{eq:analysis_mapping} about 
\begin{equation} \label{eq:analysis_VI}
	\bar u = \{u~|~\tilde F( u) + \mc{N}_\mc{X}( u)\ni 0\},
\end{equation}
before proceeding to derive bounds on the sequence of cost function values $\phi(u_k)$ generated by the algorithm.

We begin by showing that enforcing the tightened constraint set ensures state constraint satisfaction.
\begin{lmm} \label{lmm:cstr_tightening}
Given Assumptions~\ref{ass:uncertainty}-\ref{ass:convergence-assumptions} we have that
\begin{equation*}
	u \in \mc{X} \implies u\in \mc{U} \text{ and } Gu + w + d \in \mc{Y}~~\forall d\in \mc{D}.
\end{equation*}
\end{lmm}
\begin{proof}
Let $y = Gu + w$ and note that that $u \in \mc{X} \implies y_i = G_i u + w_i \in (\mc{Y} \ominus \mc{D})$ and $u\in \mc{U}$. Moreover, Assumption~\ref{ass:uncertainty} implies the existence of $\lambda_i \geq 0$ such that $\sum_{i=1}^N \lambda_i = 1$ and $[G~w] = \sum_{i=1}^N \lambda_i [G_i~w_i]$ which in turn implies that $y = \sum_{i=1}^N \lambda_i y_i \in (\mc{Y} \ominus \mc{D})$ by convexity of $(\mc{Y} \ominus \mc{D})$. Then by the definition of the Pontryagin set difference, $y \in \mc{Y}\ominus \mc{D} \implies y + d \in \mc{Y}$ for all $d\in \mc{D}$
\end{proof}

Next we focus on convergence and stability. The following preparatory lemma provides estimates of the monotonicity and Lipschitz constants of the mapping $W^{-1} \tilde F$ in \eqref{eq:analysis_mapping}.
\begin{lmm} \label{lmm:mu_and_L}
Given Assumption~\ref{ass:uncertainty}, if for all $i\in \intrng{1}{N}$ 
\begin{equation*}
	\mu_i \|x-y\|_W^2 \leq \la W^{-1} H_i (x- y), x-y\ra_W \leq L_i \|x-y\|_W^2,
\end{equation*}
where $H_i = R+ M^T Q G_i$, then 
\begin{equation*}
	\mu \|x-y\|_W^2 \leq \la W^{-1}\tilde F(x) - W^{-1}\tilde F(y),x-y\ra_W \leq L \|x-y\|_W^2.
\end{equation*}
\end{lmm}
\begin{proof}
Under Assumption~\ref{ass:uncertainty}, we have that $\tilde H\in \mathrm{co}~\{H_i,~i\in \intrng{1}{N}\}$. Moreover, $\tilde F$ is affine and thus
\begin{equation}
	\la W^{-1}\tilde F(x) - W^{-1}\tilde F(y),x-y\ra_W = \la \tilde H(x-y),x-y \ra.
\end{equation}
Starting with the lower bound:
\begin{align*}
& \la \tilde H(x-y),x-y \ra  = \textstyle \sum_{i=1}^N\lambda_i \left \la W^{-1}H_i(x-y),x-y\right \ra_W\\ 
&\geq \textstyle \sum_{i=1}^N\lambda_i \mu_i \|x-y\|_W^2 \\
&\geq \textstyle \sum_{i=1}^N\lambda_i \mu \|x-y\|_W^2 = \mu\|x-y\|_W^2
\end{align*}
The proof of the upper bound is analogous. 
\end{proof}

Next we show that the closed-loop \eqref{eq:closed-loop} satisfies constraints, is input-to-state stable, and characterizes its convergence rate. The convergence rate depends
on the smallest and largest Eigenvalues $\tilde \mu > 0$ and $\tilde L > 0$ of the matrix
\begin{equation}
	\tilde \Xi =0.5 W^{-\frac12}  (\tilde H+\tilde H^T) W^{-\frac12}, \quad \tilde H = M^T Q G + R.
\end{equation}

\begin{thm} \label{eq:thm_convergence}
Let Assumptions~\ref{ass:uncertainty}-\ref{ass:convergence-assumptions} hold. Then $u_k \in \mc{U}$ and $y_k \in \mc{Y}$ for all $k\geq 0$ and the closed-loop \eqref{eq:closed-loop} satisfies
\begin{equation*}
	\|u_{k} - \bar u\|_W \leq \eta^k \|u_0 - \bar u\|_W + \gamma \textstyle \sup_{k\geq0}\|d_k\|
\end{equation*}
where $\eta^2 = 1- (\tilde \mu/\tilde L)^2 (1-\epsilon^2)$, $\epsilon = 1 - \alpha\tilde L^2/\tilde \mu$, $\bar u$ is defined in \eqref{eq:analysis_VI}, and $\gamma = \|W^\frac12 M^T Q\|/(1-\eta)$. Moreover, $\epsilon \in (-1,1)$ and  $\eta \in (0,1)$. 
\end{thm}
\begin{proof}
Constraint satisfaction follows immediately from the projection operation in \eqref{eq:analysis_mapping} and Lemma~\ref{lmm:cstr_tightening}. 

Next we demonstrate ISS, define $A = I - \alpha W^{-1} \tilde F$. The mapping $W^{-1} \tilde F$ is affine and thus satisfies
\begin{equation*}
	\tilde\mu \|x-y\|_W^2 \leq \la W^{-1} (\tilde F(x) -\tilde F(y)), x-y\ra_W \leq \tilde L \|x-y\|_W^2.
\end{equation*}
Further as $M\in \textrm{co}~\{G_i~i\in \ints_{[1,N]}\}$, by Lemma~\ref{lmm:mu_and_L} we have that $[\tilde \mu,\tilde L]\subseteq [\mu,L]$ and thus $\tilde \mu > 0$ due to Assumption~\ref{ass:convergence-assumptions}. Hence, by \cite[Prop 25.9]{bauschke2011convex}, $A$ satisfies $\|A(x)-A(y)\|_W \leq \eta \|x-y\|_W$, with $\eta^2 = 1-2\alpha \tilde \mu + \alpha^2 \tilde L^2$, and $\eta < 1$ if $\alpha \in (0,2\tilde\mu/\tilde L^2)$.

Combining $\eta^2 = 1-2\alpha \tilde \mu + \alpha^2 \tilde L^2$ with $\alpha = \tilde \mu/\tilde L^2(1-\epsilon)$ and simplifying yields that $\eta^2 = 1 - (\tilde \mu/\tilde L)^2(1-\epsilon^2)$. Since $\alpha \in (0,2\mu/L^2)$ and $[\tilde \mu,\tilde L]\subseteq [\mu,L]$ we have that $\epsilon \in (-1,1)$ which in turn implies that $\eta\in (0,1)$.


Next, recall that $\bar u = \Pi_\mc{X}^W\circ A(\bar u)$ and thus
\begin{align}
\|u_{k+1} - \bar u\|_W &= \|\Pi_\mc{X}^W [A(u_k) + M^TQ d_k] - \bar u\|_W \nonumber\\
&= \|\Pi_\mc{X}^W [A(u_k) + M^T Qd_k] - \Pi_\mc{X}^W [A(\bar u)]\|_W\nonumber\\
&\leq \|A(u_k) - A(\bar u) + M^T Q d_k\|_W\nonumber\\
&\leq \eta \|u_k - \bar u\|_W + \|W^{\frac12} M^TQ\| \|d_k\|, \label{eq:ISS}
\end{align}
where we used non-expansiveness of the projection in the norm induced by the weighting matrix \cite[Theorem 6.42]{beck2017first} in the first inequality. Iterating \eqref{eq:ISS} yields that
\begin{align*}
	\|u_k-\bar u\|_W &\leq \eta^k \|u_0 - \bar u\|_W + \textstyle \sum_{j=0}^k \eta^{k-j} \|W^
\frac12 M^T Q\|\|d_j\|\\
	&\leq \eta^k  \|u_0 - \bar u\|_W +  \|W^
\frac12 M^T Q\| \sum_{k=0}^\infty \eta^k \sup_{k\geq 0}\|d_k\|.
\end{align*}
Since $\eta \in (0,1)$, the geometric series converges and $\|u_k-\bar u\|_W \leq \eta^k \|u_0 - \bar u\|_W + \gamma \textstyle \sup_{k\geq 0}\|d_k\|$ where $\gamma = \|W^\frac12 M^TQ\|/(1-\eta)$ as claimed.
\end{proof}

Theorem~\ref{eq:thm_convergence} gives sufficient conditions under which \eqref{eq:closed-loop} is stable about $\bar u$. We now proceed to bound the distance between the fixed point $\bar u$ of \eqref{eq:analysis_mapping} and
\begin{equation} \label{eq:opt_compact_modified}
u^* = \textstyle \argmin_{u\in \mc{X}} \phi(u) = \argmin_{u\in \mc{X}} \frac12 u^T H u + f^T u
\end{equation}
as an intermediate step in obtaining a bound on the sequence of cost function values $\phi(u_k)$.

\begin{lmm} \label{lmm:distance}
Let Assumptions~\ref{ass:uncertainty}-\ref{ass:convergence-assumptions} hold. Then $\|\bar u - u^*\|_W \leq \delta$ where $\delta = \|\tilde H\|_W \|(G-M)^T Q(Gu^* + w -r)\|_W$.
\end{lmm}
\begin{proof}
Since $\mu > 0$ by Assumption~\ref{ass:convergence-assumptions}, $\tilde H$ is strongly monotone, $\mc{A} = (\tilde H + \mc{N}_\mc{X})^{-1}$ is a Lipschitz continuous function \cite[Theorem 2F.6]{dontchev2009implicit} and $\bar u = \mc{A}(-\tilde f)$. Thus the vector $u^*$ satisfies
\begin{alignat*}{2}
	& && H u^* + f + \mc{N}_\mc{X}(u^*) \ni 0\\
	& \Leftrightarrow &&H u^* + \tilde H u^* - \tilde Hu^* + f + \mc{N}_\mc{X}(u^*)\ni 0\\
	&\Leftrightarrow u^* &&= (\tilde H + \mc{N}_\mc{X})^{-1}((\tilde H-H)u^* - f)\\
	& &&= \mc{A}((\tilde H-H)u^* - f).
\end{alignat*}
Thus, using Lipschitz continuity of $\mc{A}$,
\begin{align*}
\|\bar u - u^*\|_W &= \|\mc{A}(-\tilde f) - \mc{A}((\tilde H-H)u^* -f)\|_W\\
&\leq \|\tilde H\|_W \|(H-\tilde H)u^* + f - \tilde f\|_W\\
& =  \|\tilde H\|_W \|(G-M)^TQ(Gu^* + w -r)\|_W
\end{align*}
as claimed, where $ \|\tilde H\|_W$ is a Lipschitz constant for $\mc{A}$.
\end{proof} 
\noindent Finally, we bound the sequence of cost function values $\phi(u_k)$.
\begin{thm} \label{thm:main_result}
Let Assumptions~\ref{ass:uncertainty}-\ref{ass:convergence-assumptions} hold. Then $u_k\in \mc{U}$, $y_k \in \mc{Y}$ for all $k\geq 0$ and there exist $\beta_1 > 0$ and $\beta_2 > 0$ such that
\begin{subequations}
\begin{align}
	\|u_k - u^*\|_W &\leq \eta^k \|u_0 - \bar u\|_W + \gamma~\sup_{k\geq 0}\|d_k\| + \delta, \label{eq:final1}\\
	\phi(u) - \phi(u^*) &\leq \beta_1 \|u-u^*\|_W + \beta_2\|u - u^*\|_W^2
\end{align}
\end{subequations}
where $\eta \in (0,1)$, $\gamma > 0$ and $\delta \geq 0$ are defined in Theorem~\ref{eq:thm_convergence} and Lemma~\ref{lmm:distance}, and $\bar  u$ and $u^*$ are defined in \eqref{eq:analysis_VI} and \eqref{eq:opt_compact_modified} respectively.
\end{thm}
\begin{proof}
Combining Theorems~\ref{eq:thm_convergence} and Lemma \ref{lmm:distance} with the triangle inequality yields \eqref{eq:final1} immediately. Moreover, since $\nabla \phi(u) = Hu + f$ is Lipschitz, we have that
\begin{align*}
	\phi(u^*) - \phi(u)& \leq \la W^{-\frac12} \nabla\phi(u^*), W^{\frac12} (u^* - u)\ra +\frac{\bar L}{2}\|u-u^*\|_W^2 \\
	& \leq \|W^{-\frac12} \nabla\phi(u^*)\| \|u^* - u\|_W + \frac{\bar L}{2}\|u-u^*\|_W^2
\end{align*}
where $\bar L$ is the largest Eigenvalue of $W^{-\frac12}HW^{-\frac12}$. Letting $\beta_1 = \|W^{-\frac12} \nabla\phi(u^*)\|$ and $\beta_2 = \bar L/2$ completes the proof.
\end{proof}

\begin{rmk} \label{rmk:convergence_rate}
The convergence rate $\eta$ in Theorem~\ref{eq:thm_convergence} depends on the conditioning of $\tilde \Xi =0.5 W^{-\frac12}(\tilde H+\tilde H^T) W^{-\frac12}$. Since $\tilde H = R + M^T Q G$ and $W = R + M^T Q M$ we see that $\eta$ decreases as the model improves and $M$ approaches $G$. Similarly, we see from Lemma~\ref{lmm:distance} that the offset term $\delta$ also shrinks as the model improves. The choice of step length has an impact as well, picking $\alpha = \alpha^* = \tilde\mu/\tilde L^2$ minimizes $\eta$. However $\tilde\mu$ and $\tilde L$ are unknown and $\epsilon = \tilde L^2 /\tilde \mu(\alpha^*  - \alpha)$ quantifies the loss of performance due to suboptimal selection of $\alpha$.
\end{rmk}








\section{Numerical Example} \label{ss:simulations}
We illustrate the potential of our OB-ILC methodology with an application to feedforward control of a high-precision two-axis robotic motion stage. The control architecture for each axis is shown in Figure~\ref{fig:block-diagram}. The objective is to track a reference trajectory $\bar p$ by driving $e$ to zero, while satisfying constraints on the speed $v$ and voltage input $V$.

\begin{figure}[htbp]
	\centering
	\includegraphics[width=0.95\columnwidth]{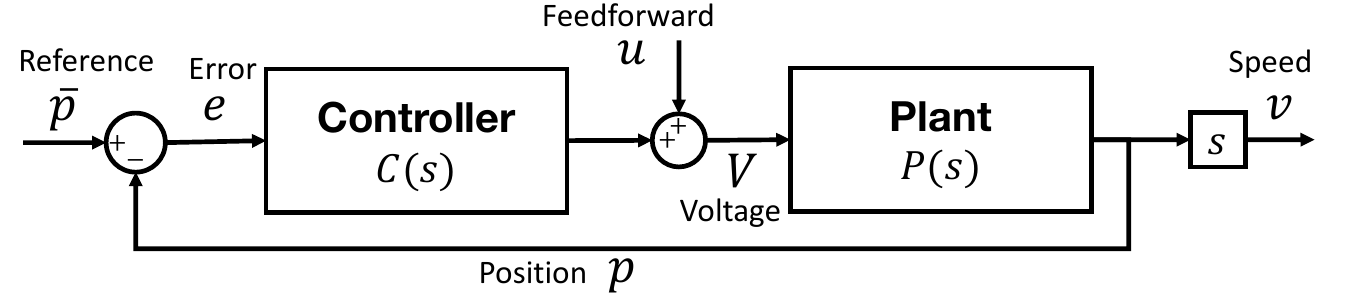}
	\caption{The control architecture of the precision motion stage. The error dynamics of each axis is stabilized with a feedback controller to track the reference $\bar p$ while ILC is used to generate a feedforward input $u$.}
	\label{fig:block-diagram}
\end{figure}

The IO map can be expressed using the transfer functions 
\begin{subequations} \label{eq:system_model}
\begin{align}
e_i(s) &= \frac{-P_i}{1 + P_i C_i}u_i(s) + \frac{1}{1 + P_i C_i} \bar p_i(s) \\
v_i(s) &= \frac{s P_i}{1 + P_i C_i}u_i(s) + \frac{s P_i C_i}{1 + P_i C_i} \bar p_i(s) \\
V_i(s) &= \frac{1}{1 + P_i C_i}u_i(s)  + \frac{C_i}{1 + P_i C_i} \bar p_i(s)
\end{align}
\end{subequations}
for axes $i \in \{x,y\}$ where $P_i(s)$ and $C_i(s)$ model a robotic deposition machine and are taken from \cite[Equations (28)-(31)]{bristow2006high} which identifies them from experimental data. The system model \eqref{eq:system_model} can be written in the compact form
\begin{equation}
	y(s) = G(s)u(s) + H(s)\bar{p}(s),
\end{equation}
with $y^T = [e_x~v_x~V_x~e_y~v_y~V_y]$, $u^T = [u_x~u_y]$, $\bar p^T = [\bar p_x~\bar p_y]$.

We obtain a finite-dimensional IO mapping by taking $N = 400$ evenly spaced samples of the trajectories generated by \eqref{eq:system_model} resulting in a sampling period of $2$ ms, $m = 2406$, and $n= 802$. The matrix representation $G$ of \eqref{eq:system_model} is an exact first-order hold discretization of a continuous time state space realization of $G(s)$. The constant term $w$ is the sampled time-response of $H(s)$ to the reference trajectory $\bar p$ in Figure~\ref{fig:modelling_error}.

\begin{figure}[htbp]
	\centering
	\includegraphics[width=0.95\columnwidth, trim=0cm 0cm 0cm 0cm, clip=true]{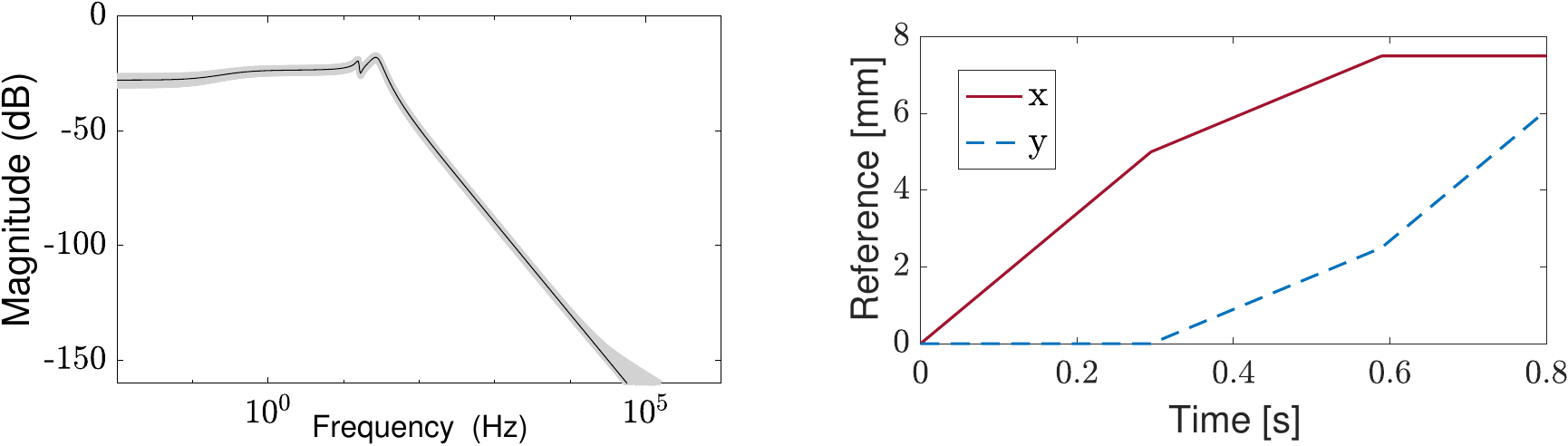}
	\caption{(Left) The frequency response of $G_{11}(s)$ including the true system (black) and the convex hull uncertainty representation (grey shading). (Right) The reference trajectory $\bar p$ used in this paper.}
	\label{fig:modelling_error}
\end{figure}

The convex hull uncertainty representation $\Theta$ in Assumption~\ref{ass:uncertainty} is generated by discretizing the four corners of
\begin{equation}
	G(s) = (1 \pm 0.1) G_s(s) + (1 \pm 0.15) G_f(s)
\end{equation}
where $G_f(s) + G_s(s) = G(s)$ is a fast/slow decomposition of $G(s)$ around $10$ rad/s obtained using \texttt{freqsep} in MATLAB. The resulting uncertain $G_{11}(s)$ is illustrated in Figure~\ref{fig:modelling_error}. The same procedure applied to $H(s)$ is used to produce $w_i$.

\begin{figure}[htbp]
	\centering
	\includegraphics[width=0.95\columnwidth, trim=0cm 0cm 0cm 0cm, clip=true]{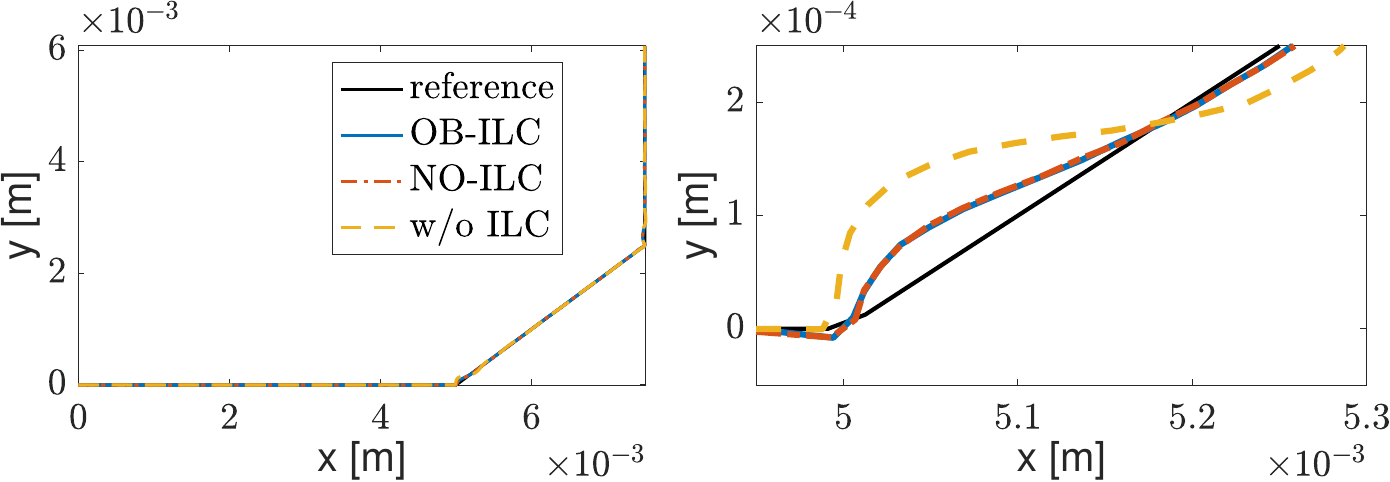}
	\caption{The inputs computed by bothfully converged} ILC algorithms are able to significantly reduce the tracking error.
	\label{fig:planar}
\end{figure}

\begin{figure}[htbp]
	\centering
	\includegraphics[width=0.95\columnwidth, trim=0cm 0cm 0cm 0cm, clip=true]{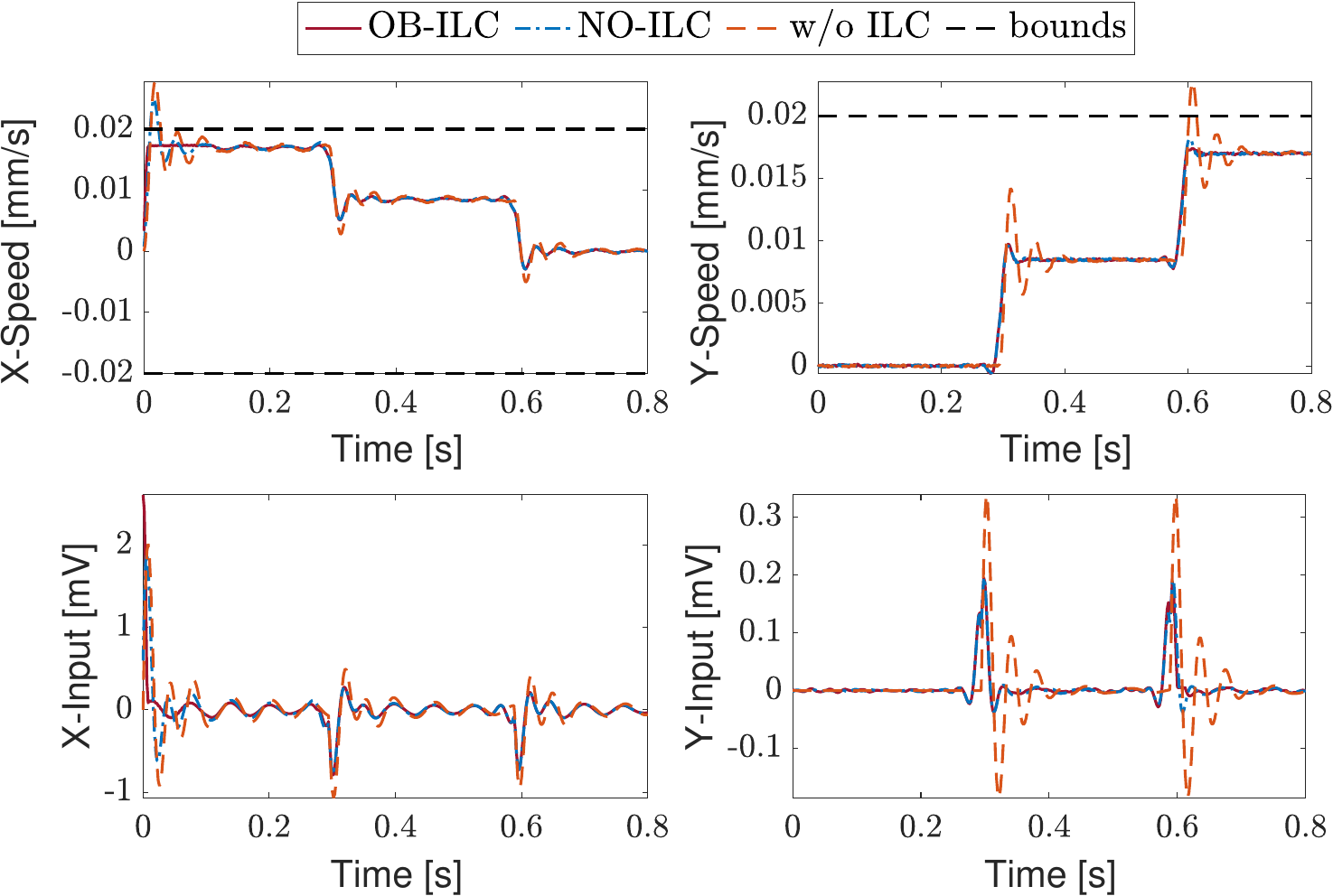}
	\caption{Our algorithm (OB-ILC) robustly satisfies the constraints while norm-optimal ILC (NO-ILC) and the plant without ILC feedforward (w/o ILC) violate the speed constraint. Both ILC algorithms have converged.}
	\label{fig:timeseries}
\end{figure}

The optimization problem is to minimize a weighted combination of the error and the control effort subject to constraints on the maximum velocity and input. The cost matrices are $Q = I_{N+1} \otimes 
\begin{bmatrix} 
	 	100 I_2 & 0\\0 & 0
	 	\end{bmatrix}$ 
and $R = I_{2(N+1)}$ where $\otimes$ is the Kronecker product, the constraints are $\mc{Y} = \{(e,v,V)~|~\|v\|_\infty \leq 0.02~\mathrm{mm/s},~ \|V\|_\infty \leq 0.1~\mathrm{V}\}$ and $\mc{U} = \reals^n$. The disturbance set is $\mc{D} = \{(e,v,V)~|~\|e\|_\infty \leq 10^{-3}~\mathrm{mm}, \|v\|_\infty \leq 10^{-4}~\mathrm{mm/s}\}$. The model $M$ is chosen randomly from $\Theta$ and in all cases the step length is $\alpha = \mu/L^2$. Our MATLAB 2020b implementation using \texttt{quadprog}'s active set algorithm running on a 2020 Macbook Pro (8-Core Apple M1) took $1.57\pm0.8$ s, with a maximum of $2.93$ s to execute. 

We compare our proposed OB-ILC method with the norm-optimal ILC (NO-ILC) controller proposed in \cite{barton2010norm}. The NO-ILC controller uses the same weighting matrices as OB-ILC and coincides with our method if $\alpha = 1$ and the constraints are removed. Both methods are able to significantly reduce the tracking error as illustrated in Figure~\ref{fig:planar}. Figure~\ref{fig:timeseries} illustrates that OB-ILC is able to enforce output constraints at the fixed point while NO-ILC cannot. This is fundamental for practical applications. 

\begin{figure}[htbp]
	\centering
	\includegraphics[width=0.95\columnwidth]{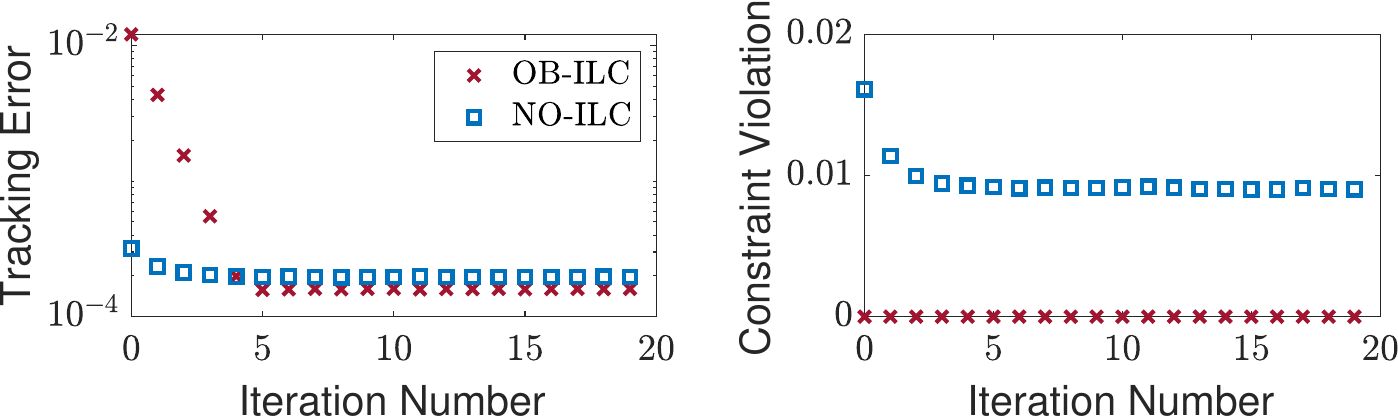}	\caption{OB-ILC converges more slowly than NO-ILC, but achieves a lower asymptotic tracking error and satisfies output constraints. Tracking error and constraint violation are quantified using the metrics $\|e\|_2$ and $\|y - \Pi_\mc{Y}(y)\|_2$.}
	\label{fig:convergence}
\end{figure}

Figure~\ref{fig:convergence} shows that OB-ILC converges more slowly than NO-ILC, but attains a lower asymptotic tracking error and satisfies the output constraints throughout. The lower cost attained by the OB-ILC is explained by the conservatism introduced through the constraint tightening, which forces the controller to be less aggressive and more robust in the presence of model mismatch especially when the constraints are active (see Fig.~\ref{fig:timeseries}). The slower convergence is expected due to the step size restriction (Assumption~\ref{ass:convergence-assumptions}.(ii)) used to ensure robust stability rather than through the selection of the  matrices $Q$ and $R$ as is common in the literature. A similar ILC approach with adaptive step-size was presented in \cite{son2015robust} but constraints were not considered in that work.



Finally, we investigate the sensitivity of the algorithm to modelling error and disturbances by varying the quality of the model $M$ for a fixed $\mc{D}$ (Figure~\ref{fig:sens_combined}, left) and by varying $\mathrm{rad}~\mc{D}$ for a fixed $M$ (Figure~\ref{fig:sens_combined}, right).
\begin{figure}[htbp]
	\centering
	\includegraphics[width=0.9\columnwidth, trim=0cm 0cm 0cm 0cm, clip=true]{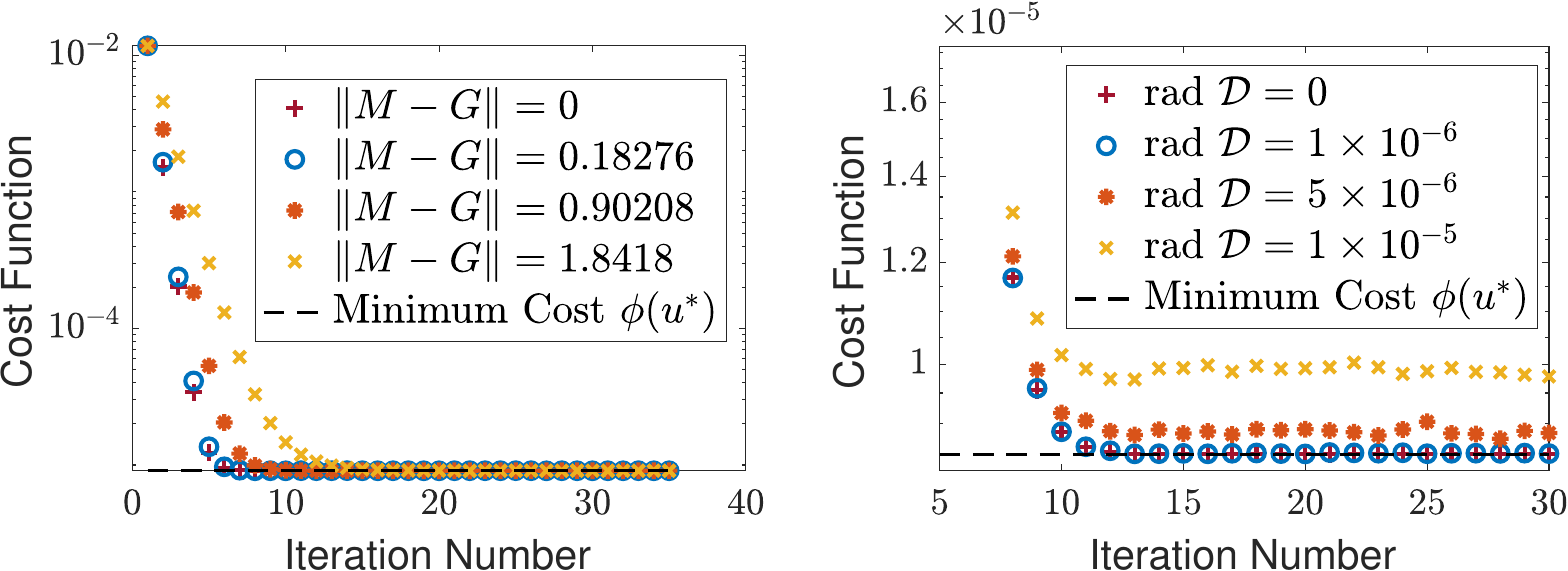}
	\caption{The algorithm converges more slowly as the modelling error increases (left) as predicted in Remark~\ref{rmk:convergence_rate}. The cost function decreases until it reaches a threshold which increases with the noise as predicted by Theorem~\ref{thm:main_result} (right).}
	\label{fig:sens_combined}
\end{figure}
\section{Conclusions}
This paper presented a robust ILC policy for state and input constrained systems in the presence of process noise and model mismatch along with a rigorous iteration domain stability analysis. It illustrates how monotone operator and robust control concepts can be combined to design and analyze ILC policies. Future work includes incorporating sensitivity and model learning to reduce conservatism and accelerate convergence.

\bibliography{opt_ilc}
\end{document}